\documentclass[graybox]{svmult}

\usepackage{mathptmx}       
\usepackage{helvet}         
\usepackage{courier}        
\usepackage{type1cm}        

\usepackage{makeidx}         
\usepackage{graphicx}        
\usepackage{multicol}        
\usepackage[bottom]{footmisc}

\makeindex             

\usepackage{amsmath,amssymb,calrsfs}
\usepackage[mathscr]{eucal}
\usepackage[hang]{subfigure} 
\usepackage{tikz,mathtools}
\usepackage{enumerate} 
\usepackage{color}
\usepackage{float}

\begin{document}

\title*{Copula-based piecewise regression}

\author{Arturo Erdely}

\institute{Arturo Erdely \at Universidad Nacional Aut\'onoma de M\'exico, Facultad de Estudios Superiores Acatl\'an, \email{arturo.erdely@comunidad.unam.mx}}

\maketitle

\abstract{Most common parametric families of copulas are totally ordered, and in many cases they are also positively or negatively regression dependent and therefore they lead to monotone regression functions, which makes them not suitable for dependence relationships that imply or suggest a non-monotone regression function. A gluing copula approach is proposed to decompose the underlying copula into totally ordered copulas that combined may lead to a non-monotone regression function.}

\section{Introduction}
\label{sec:Intro}

Given a bivariate random vector $(X,Y)$ with joint probability distribution function $F_{X,Y}(x,y)=\mathbb{P}(X\leq x,Y\leq y)$ it is possible to assess uncertainty about one of the random variables conditioning on certain values of the other, for example through the univariate conditional probability distribution of $Y$ given $X=x,$ that is $F_{Y|X}(y\,|\,x)=\mathbb{P}(Y\leq y\,|\,X=x).$ As a point estimate for a future value of $Y$ given $X=x$ we may calculate central tendency measures with $F_{Y|X}$ such as the mean (whenever it exists) or the median (which always exists in the continuous case) which will depend on the conditioning value $x$ and therefore such point estimates depending on $x$ may be denoted by $\mu(x)$ and are called \textit{regression function} for $Y$ given $X=x.$

\medskip

As a consequence of \textit{Sklar's Theorem} \cite{Skl59} for continuous random variables there exists a unique copula $C$ such that the joint probability distribution function $F_{X,Y}(x,y)=C(F_X(x),F_Y(y))$ where $F_X(x)=\mathbb{P}(X\leq x)$ and $F_Y(y)=\mathbb{P}(Y\leq y)$ are the marginal probability distribution functions of $X$ and $Y,$ respectively. As explained in \cite{Nel06}, the conditional distribution of $Y$ given $X=x$ can be obtained by

\begin{equation}
  F_{Y|X}(y\,|\,x)\,=\,\frac{\partial C(u,v)}{\partial u}\Big|_{u\,=\,F_X(x)\,,\,v\,=\,F_Y(y)}
\label{eq:condcdf}
\end{equation}
and therefore to find the median regression function for $Y$ given $X=x$ whenever $F_{Y|X}$ is a continuous distribution function, we proceed as follows:

\medskip

\noindent\textbf{Algorithm 1}
\begin{enumerate}
  \item Set $\partial C(u,v)/\partial u\,=\,1/2\,;$
	\item solve for the regression function of $V=F_Y(Y)$ given $U=F_X(X)=u,$ and obtain $v=\psi(u)\,;$
	\item replace $u$ by $F_X(x)$ and $v$ for $F_Y(y)\,;$
	\item solve for the regression function of $Y$ given $X=x:$
	   \begin{equation}
		   y\,=\,\mu(x)\,=\,F^{(-1)}_Y(\psi(F_X(x))).
		 \label{eq:regcurve}
		 \end{equation}
\end{enumerate}
It is worth to notice that since $F_X$ and $F_Y$ only explain the individual (marginal) probabilistic behavior of the continuous random variables $X$ and $Y,$ respectively, then the information about their dependence for regression purposes is contained in $\psi.$ A survey of copula-based regression models may be found in \cite{KolPai09} and estimation/inference procedures for such purpose in \cite{NohGhoBou13}.

\section{Piecewise monotone regression}
\label{sec:Piecewise}

In \cite{DetHecVol14} it is argued that when the regression function is non-monotone, copula-based regression estimates do not reproduce the qualitative features of the regression function under commonly used parametric copula families. This occurs because very often such parametric copulas lead to monotone regression functions, but in case there is evidence that the underlying regression function is non-monotone a \textit{piecewise regression} approach may be applied in order to break up a non-monotone relationship into a piecewise monotonic one, and then seek for the best copula fit for each piece. 

\medskip

Piecewise (or segmented) monotone regression for $Y$ given $X=x$ is defined by partitioning the support of $X$ into a finite number of intervals such that restricted to each one it is possible to obtain a monotone regression function. For example, instead of (\ref{eq:condcdf}) we may obtain something like
\begin{equation}
  F_{Y|X}(y\,|\,x)\,=\,\begin{cases}
	                     \,\frac{\partial}{\partial u}C_1(u,v)\Big|_{u\,=\,F_{X|X\leq\,b}(x)\,,\,v\,=\,F_Y(y)}, & \text{ if } x\leq b, \\
							  		   \,\frac{\partial}{\partial u}C_2(u,v)\Big|_{u\,=\,F_{X|X>\,b}(x)\,,\,v\,=\,F_Y(y)},& \text{ if } x > b,
									   \end{cases}
\label{eq:condcdf2}
\end{equation}
with $C_1$ and $C_2$ two different copulas, $b$ is called a \textit{break-point} for explanatory variable $X,$ and where $F_{X|X\leq\,b}$ and $F_{X|X>\,b}$ are the conditional distribution functions of $X$ given $X\leq b$ and $X>b,$ respectively. This may be justified in terms of the \textit{gluing copula} technique \cite{SibSto08} as explained in \cite{ErdDia10} for the particular case of vertical section gluing and bivariate copulas. Specifically, given two bivariate copulas $C_1$ and $C_2,$ and a fixed value $0<\theta<1$ (gluing point), we may scale $C_1$ to $[0,\theta]\times[0,1]$ and $C_2$ to $[\theta,1]\times[0,1]$ and \textit{glue} them into a single copula
\begin{equation}
  C_{1,2,\theta}(u,v)\,=\,\begin{cases}
	                          \,\theta C_1(\frac{u}{\theta},v), & 0\leq u\leq\theta,\\  
														\,(1-\theta)C_2(\frac{u-\theta}{1-\theta},v)+\theta v, & \theta\leq u\leq 1.
	                        \end{cases}
\label{eq:gluing}
\end{equation}
Then
\begin{equation}
  \frac{\partial}{\partial u}C_{1,2,\theta}(u,v)\,=\,\begin{cases}
	                                                     \,\frac{\partial}{\partial u}C_1(\frac{u}{\theta},v), & 0\leq u\leq\theta,\\  
														                           \,\frac{\partial}{\partial u}C_2(\frac{u-\theta}{1-\theta},v), & \theta\leq u\leq 1,
	                                                   \end{cases}
\label{eq:gluingdu}
\end{equation}
and by (\ref{eq:condcdf})
\begin{eqnarray}
  F_{Y|X}(y\,|\,x) &=& \frac{\partial C_{1,2,\theta}(u,v)}{\partial u}\bigg|_{u\,=\,F_X(x)\,,\,v\,=\,F_Y(y)} \nonumber \\
	                 &=& \begin{cases}
									       \,\frac{\partial}{\partial u}C_1(\frac{F_X(x)}{\theta},F_Y(y)), & 0\leq F_X(x)\leq\theta,\\  
												 \,\frac{\partial}{\partial u}C_2(\frac{F_X(x)-\theta}{1-\theta},F_Y(y)), & \theta\leq F_X(x)\leq 1,
									     \end{cases} \nonumber \\
				&=& \begin{cases}
							  \,\frac{\partial}{\partial u}C_1(u,v)\Big|_{u\,=\,F_{X|X\leq\,b}(x)\,,\,v\,=\,F_Y(y)}, & x\leq F^{(-1)}_X(\theta)=b, \\ 
							  \,\frac{\partial}{\partial u}C_2(u,v)\Big|_{u\,=\,F_{X|X>\,b}(x)\,,\,v\,=\,F_Y(y)}, & x > F^{(-1)}_X(\theta)=b,
				    \end{cases}
\label{eq:condcdf3}
\end{eqnarray}
since $F_{X|X\leq\,b}(x)=\mathbb{P}(X\leq x\,|\,X\leq b) = \mathbb{P}(X\leq x)/\mathbb{P}(X\leq b) = F_X(x)/\theta$ and\linebreak $F_{X|X>\,b}(x) = \mathbb{P}(b<X\leq x)/\mathbb{P}(X>b) = (F_X(x)-\theta)/(1-\theta).$ The result obtained in (\ref{eq:condcdf3}) leads to a regression function of the form
\begin{equation}
  \mu(x)\,=\,\begin{cases}
	                   \,\mu_1(x), & \text{ if } x\leq b, \\
							       \,\mu_2(x), & \text{ if } x > b,
	                 \end{cases}
\label{eq:regcurve2}
\end{equation}
where, for example, if $\mu_1(x)$ is an increasing function and $\mu_2(x)$ a decreasing one, then $\mu(x)$ is non-monotone.

\begin{example}\label{ex:NelsenEj33}
  From example 3.3 in \cite{Nel06} if a probability mass $0<\theta<1$ is uniformly distributed on the line segment joining $(0,0)$ to $(\theta,1),$ and a probability mass $1-\theta$ is uniformly distributed on the line segment joining $(\theta,1)$ to $(1,0),$ see Fig. \ref{fig:Example1}, the underlying copula for a random vector $(X,Y)$ of continuous Uniform$(0,1)$ random variables with such non-monotone dependence is given by
	\begin{equation}
	  C_{\theta}(u,v)\,=\,\begin{cases}
		                      \,u, & 0\leq u\leq\theta v\leq \theta, \\
													\,\theta v, & 0\leq\theta v<u<1-(1-\theta)v, \\
													\,u+v-1, & \theta\leq 1-(1-\theta)v\leq u\leq 1.
		                    \end{cases}
	\label{eq:NelsenEj33A}
	\end{equation}
	
  \begin{figure}[t]
  \sidecaption
  \includegraphics[scale=.30]{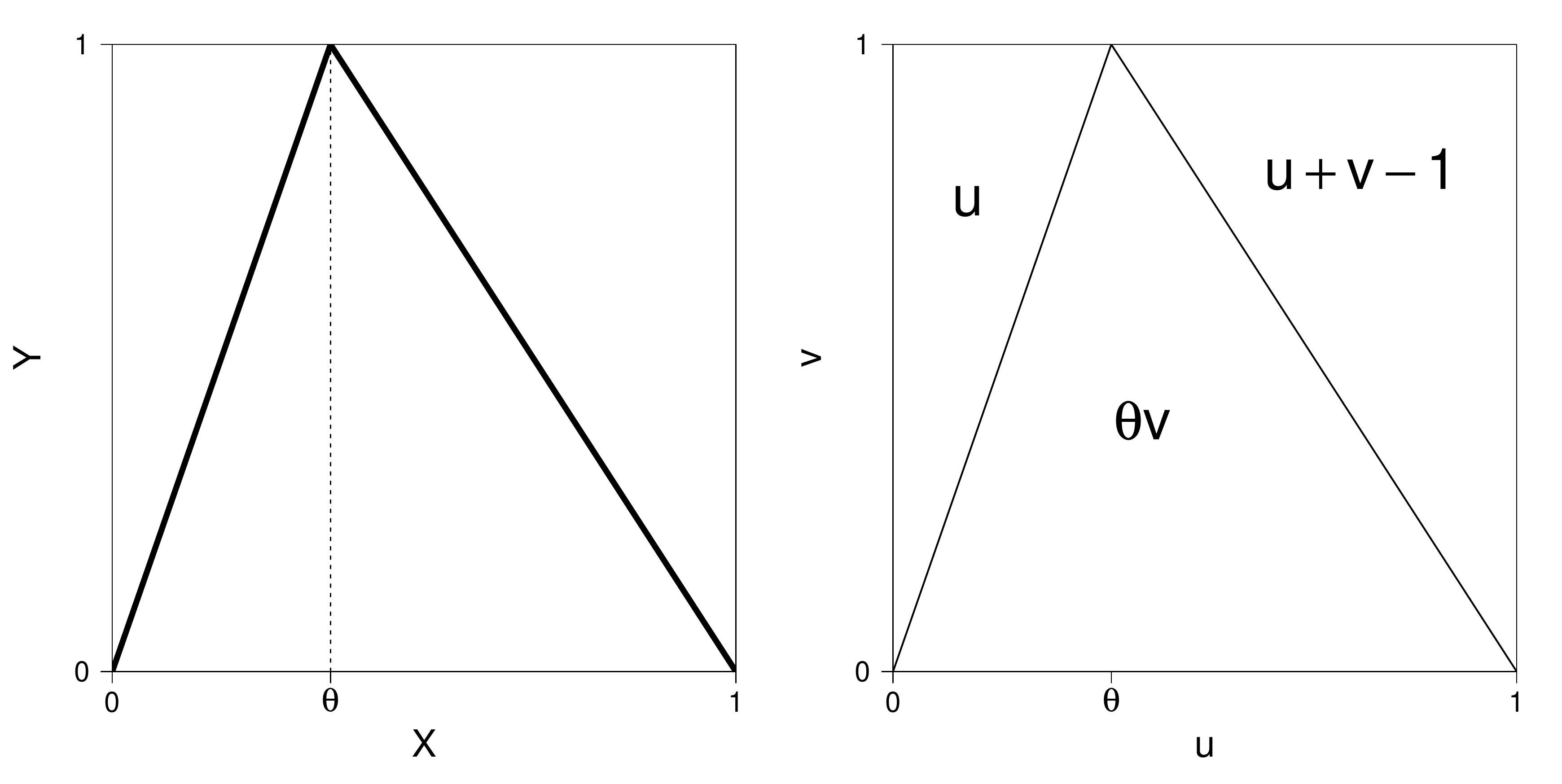}
  \caption{Example \ref{ex:NelsenEj33}. Left: $(X,Y)$ dependence. Right: underlying copula (\ref{eq:NelsenEj33A}).}
  \label{fig:Example1} 
  \end{figure}
	
	By construction we have that $\mathbb{P}(Y=\frac{x}{\theta}\,|\,X=x)=1$ whenever $0\leq x\leq\theta$ and $\mathbb{P}(Y=\frac{1-x}{1-\theta}\,|\,X=x)=1$ whenever $\theta\leq x\leq 1,$ which implies that the regression function of $Y$ given $X=x$ is
	\begin{equation}
	  \mu(x)\,=\,\begin{cases}
		                   \,\,\,\,\frac{x}{\theta}, & 0\leq x\leq\theta, \\
											 \,\frac{1-x}{1-\theta}, & \theta\leq x\leq 1,
		                 \end{cases}
	\label{eq:NelsenEj33B}
	\end{equation}
clearly a non-monotone function: linearly increasing for $0\leq x\leq\theta$ and linearly decreasing for $\theta\leq x\leq 1,$ which suggests in this case that the underlying dependence might be split by means of the gluing copula technique in terms of two copulas, with $\theta$ as gluing point. Indeed, let $C_1(u,v)=\min\{u,v\}$ (the Fr\'echet-Hoeffding upper bound that represents the case when one variable is almost surely an increasing function of the other) and $C_2(u,v)=\max\{u+v-1,0\}$ (the Fr\'echet-Hoeffding lower bound that represents the case when one variable is almost surely a decreasing function of the other), then applying (\ref{eq:gluing}) it is straightforward to verify that the resulting gluing copula $C_{1,2,\theta}$ is equal to (\ref{eq:NelsenEj33A}).
	
	\medskip
	
	Therefore, the same regression function obtained in (\ref{eq:NelsenEj33B}) could be obtained in two pieces: the first one in terms of the random vector $(X_1,Y)$ with underlying copula $C_1$ and where the distribution of $X_1$ is the conditional distribution of $X$ given $X\leq\theta,$ which turns to be uniform$(0,\theta),$ and the second one in terms of the random vector $(X_2,Y)$ with underlying copula $C_2$ and where the distribution of $X_2$ is the conditional distribution of $X$ given $X>\theta,$ which turns to be uniform$(\theta, 1).$ Applying (\ref{eq:condcdf}) to the first piece we obtain the following:
	\begin{eqnarray}
	  F_{Y|X_1}(y\,|\,x) &=& \frac{\partial}{\partial u}C_1(u,v)\Big|_{u\,=\,\frac{x}{\theta}\,,\,v\,=\,y} \nonumber \\
		                   &=& \begin{cases}
											       \,1, & \text{ if } y\geq\frac{x}{\theta}\,, \\
														 \,0, & \text{ if } y < \frac{x}{\theta}
											     \end{cases}
	\label{eq:NelsenEj33C}
	\end{eqnarray}
	from which we get $\mu_1(x)=\frac{x}{\theta}$ whenever $0\leq x\leq\theta,$ and similarly from $F_{Y|X_2}(y\,|\,x)$ we obtain $\mu_2(x)=\frac{1-x}{1-\theta}$ whenever $\theta\leq x\leq 1,$ as expected.$\qquad_{\blacksquare}$ 
\end{example}

For simplicity's sake, the case for a single break-point has been analyzed, but the analogous idea may be applied for finitely many break-points. For each interval $I$ induced in the support of the explanatory variable, the conditional distribution of $Y$ given $X=x$ is obtained by
\begin{equation}
  F_{Y|X}(y\,|\,x)\,=\,\frac{\partial}{\partial u}C_I(u,v)\Big|_{u\,=\,F_{X|X\in\,I}(x)\,,\,v\,=\,F_Y(y)}
\label{eq:condcdf4}
\end{equation}
and with it the regression function $\mu(x)$ for $x\in I$ may be calculated.

\section{Dependence and regression}
\label{sec:DepReg}

In this section the concepts of quadrant and regression dependence by \cite{Leh66} are recalled.

\begin{definition}\label{def:PQD}
  A bivariate random vector $(X,Y)$ or its joint distribution function $F_{X,Y}$ is \textit{positively quadrant dependent} and abbreviated as $\text{PQD}(X,Y)$ if
	\begin{equation}
	  \mathbb{P}(X\leq x, Y\leq y)\,\geq\,\mathbb{P}(X\leq x)\mathbb{P}(Y\leq y)\,,\qquad \text{for all }x \text{ and } y,
	\label{eq:PQD}
	\end{equation}
	and \textit{negatively quadrant dependent} $\text{NQD}(X,Y)$ if (\ref{eq:PQD}) holds with the inequality sign reversed.
\end{definition}

In the particular case where both $X$ and $Y$ are continuous random variables with underlying copula $C,$ as an immediate consequence of \textit{Sklar's Theorem} \cite{Skl59} we have that $\text{PQD}(X,Y)$ is equivalent to $C(u,v)\geq uv$ for all $u,v$ in $[0,1],$ and 
$\text{NQD}(X,Y)$ with this last inequality sign reversed. From \cite{Nel06} we have the following:

\begin{definition}\label{def:order}
  If $C_1$ and $C_2$ are copulas, we say that $C_1$ is \textit{smaller than} $C_2$ (or $C_2$ \textit{is larger than} $C_1$), and write $C_1\prec C_2$ (or $C_2\succ C_1$) if $C_1(u,v)\leq C_2(u,v)$ for all $u,v$ in $[0,1].$ 
\end{definition}

This point-wise partial ordering of the set of copulas is called \textit{concordance ordering.} It is a partial order rather than a total order because not every pair of copulas is comparable. However, there are families of copulas that are totally ordered. We will call a totally ordered parametric family $\{C_{\theta}\}$ of copulas \textit{positively ordered} if $C_{\alpha}\prec C_{\beta}$ whenever $\alpha\leq\beta;$ and \textit{negatively ordered} if $C_{\alpha}\succ C_{\beta}$ whenever $\alpha\leq\beta.$ Many of well known one-parameter families of copulas are totally ordered and include $\Pi(u,v)=uv,$ and hence have subfamilies of PQD and NQD copulas.

\medskip

As mentioned in \cite{Nel06} one form to calculate \textit{Spearman's concordance measure} is
\begin{equation}
  \rho_{\,C}\,=\,12\int\!\!\!\!\int_{[0,1]^2}\big[\,C(u,v)-uv\,\big]\,dudv\,=\,12\int\!\!\!\!\int_{[0,1]^2}C(u,v)\,dudv\,-\,3\,,
\label{eq:Spearman}
\end{equation}
and hence $\rho_{\,C}/12$ can be interpreted as a measure of ``average'' quadrant dependence (both positive and negative) for continuous random variables whose copula is $C.$ Closely related to (\ref{eq:Spearman}) is the $L_1$ distance between $C$ and the (sometimes called) independence copula $\Pi(u,v)=uv$ known as \textit{Schweizer-Wolff's dependence measure} \cite{SchWol81} defined as
\begin{equation}
  \sigma_{\,C}\,=\,12\int\!\!\!\!\int_{[0,1]^2}\big|C(u,v)-uv\big|\,dudv\,. 
\label{eq:Schweizer}
\end{equation}
Two main differences (among others) are that $-1\leq\rho_{\,C}\leq 1$ in contrast to $0\leq\sigma_{\,C}\leq 1,$ and that $\sigma_{\,C}=0$ if and only $X$ and $Y$ are independent (that is $C=\Pi$) while $\rho_{\,C}=0$ does not necessarily imply independence. Moreover, as explained in \cite{Nel06}:
\begin{quote}\label{qu:Nelsen}
  Of course, it is immediate that if $X$ and $Y$ are PQD, then $\sigma_{\,X,Y}=\rho_{\,X,Y}\,;$ and that if $X$ and $Y$ are NQD, then $\sigma_{\,X,Y}=-\rho_{\,X,Y}\,.$ Hence for many of the totally ordered families of copulas presented in earlier chapters (e.g., Plackett, Farlie-Gumbel-Morgenstern, and many families of Archimedean copulas), $\sigma_{\,X,Y}=|\rho_{\,X,Y}|.$ But for random variables $X$ and $Y$ that are neither PQD nor NQD, i.e., random variables whose copulas are neither larger nor smaller than $\Pi,$ $\sigma$ is often a better measure than $\rho$ [\ldots] 
\end{quote}

\begin{definition}\label{def:PRD}
  A random variable $Y$ is \textit{positively regression dependent} on a random variable $X$ and abbreviated as $\text{PRD}(Y|X)$ if
	\begin{equation}
	  F_{Y|X}(y\,|\,x)\,=\,\mathbb{P}(Y\leq y\,|\,X=x)\quad\text{is non-increasing in } x,
	\label{eq:PRD}
	\end{equation}
	and \textit{negatively regression dependent} $\text{NRD}(Y|X)$ if (\ref{eq:PRD}) is non-decreasing in $x.$
\end{definition}

From theorems 5.2.4 and 5.2.12 in \cite{Nel06} or from Lemma 4 in \cite{Leh66} we have the following:

\begin{corollary}
  Given $(X,Y)$ a bivariate random vector:
  \begin{itemize}
	  \item[a)] If $\,\text{PRD}(Y|X)\,$ then $\,\text{PQD}(X,Y).$
		\item[b)] If $\,\text{NRD}(Y|X)\,$ then $\,\text{NQD}(X,Y).$
	\end{itemize}
\label{cor:QDimpliesRD}
\end{corollary}
By arguments explained in \cite{Nel06} the reverse implications in Corollary \ref{cor:QDimpliesRD} do not necessarily hold. 

\begin{corollary}
  If $(X,Y)$  are continuous random variables with underlying copula $C$ then:
	\begin{itemize}
	  \item[a)] $\text{PRD}(Y|X)$ if and only if for any $v$ in $[0,1]$ and for almost all $u,$ $\partial C(u,v)/\partial u$ is non-increasing in $u;$
		\item[b)] $\text{NRD}(Y|X)$ if and only if for any $v$ in $[0,1]$ and for almost all $u,$ $\partial C(u,v)/\partial u$ is non-decreasing in $u.$
	\end{itemize}
\label{cor:NelsenThm5210}
\end{corollary}

In case the conditional expectation exists it is possible to obtain a \textit{mean regression function}
\begin{equation}
  \mu(x) \,=\, \mathbb{E}(Y\,|\,X=x) \,=\, \int_0^{\,\infty}[1-F_{Y|X}(y\,|\,x)]\,dy\,-\,\int_{-\infty}^{\,0} F_{Y|X}(y\,|\,x)\,dy\,,
\label{eq:regmedia}
\end{equation}
and in case $F_{Y|X}(y\,|\,x)$ is a continuous function of $y$ then it is possible to obtain a \textit{median regression function}
\begin{equation}
  \mu(x) \,=\,\text{median}(Y\,|\,X=x) \,=\, F_{Y|X}^{(-1)}(0.5\,|\,x)\,.
\label{eq:regmediana}
\end{equation}

\begin{proposition}
  Let $\mu(x)$ be a mean or median regression function:
  \begin{itemize}
	  \item[a)] If $\text{PRD}(Y|X)$ then $\mu(x)$ is non-decreasing.
		\item[b)] If $\text{NRD}(Y|X)$ then $\mu(x)$ is non-increasing.
	\end{itemize}
\label{prop:regresion}
\end{proposition}
\begin{proof}
  If $\text{PRD}(Y|X)$ then for all $x_1<x_2$
  \begin{eqnarray}
	  -F_{Y|X}(y\,|\,x_1) &\leq& -F_{Y|X}(y\,|\,x_2)\,, \label{eq:PRD2} \\
	  1 - F_{Y|X}(y\,|\,x_1) &\leq& 1 - F_{Y|X}(y\,|\,x_2)\,. \label{eq:PRD3}
	\end{eqnarray}
	Integration of (\ref{eq:PRD2}) on $\,]-\infty,0]$ and of (\ref{eq:PRD3}) on $[0,\infty[\,,$ and adding the results according to the inequalities it is obtained $\mu(x_1)=\mathbb{E}(Y\,|\,X=x_1)\leq\mathbb{E}(Y\,|\,X=x_2)=\mu(x_2),$ as required. Now from (\ref{eq:PRD2}) we have $F_{Y|X}(y\,|\,x_1)\geq F_{Y|X}(y\,|\,x_2),$ and since $F_{Y|X}(y\,|\,x)$ is non-decreasing in $y$ for any $x$ so is $F_{Y|X}^{(-1)}(u\,|\,x)$ as a function of $u$ and therefore $F_{Y|X}^{(-1)}(u\,|\,x_1)\leq F_{Y|X}^{(-1)}(u\,|\,x_2),$ hence $\mu(x_1)=\text{median}(Y\,|\,X=x_1)=\,$$F_{Y|X}^{(-1)}(0.5\,|\,x_1)\leq\,$$F_{Y|X}^{(-1)}(0.5\,|\,x_2)=\,$$\text{median}(Y\,|\,X=x_2)=\mu(x_2),$ as required.$\qquad_{\blacksquare}$ 
\end{proof}

But the reverse implications in this last proposition do not necessarily hold, as it can be easily verified by similar arguments.

\begin{example}\label{ex:NelsenEj33bis}
  Continuing with Example \ref{ex:NelsenEj33}, applying formulas (\ref{eq:Spearman}) and (\ref{eq:Schweizer}) it is straightforward to verify that Spearman's $\rho_{\theta}=2\theta-1$ and Schweizer-Wolff's $\sigma_{\theta}=\theta^2+(\theta-1)^2,$ and since $0<\theta<1$ then $|\rho_{\theta}|<\sigma_{\theta}$ and therefore neither we have PQD nor NQD, and neither PRD nor NRD. Moreover, if $\theta=\frac{1}{2}$ then $\rho_{1/2}=0$ but this does not imply independence since $\sigma_{1/2}=\frac{1}{2}$ (its minimum possible value, by the way). See Fig. \ref{fig:Examples23} (left).$\qquad_{\blacksquare}$
\end{example}

\section{Change-point detection}
\label{sec:Chngpt}

The ideas explained in the previous sections may be useful in tackling the concerns raised by \cite{DetHecVol14} when the dependence relationship between random variables implies a non-monotone regression function, considering that the most common families of parametric copulas lead to monotone regression functions, and a possible solution might be to break up such dependence into pieces such that within each one the dependence implies a piecewise monotone regression function, and possibly one of the common families of parametric copulas may have an acceptable fit for each piece. In pursuing this objective, when dealing with data from which the dependence has to be estimated, a methodology to find break-point candidates, that is \textit{change-point detection,} becomes necessary. 

\medskip

\begin{definition}\label{def:diagonal}
  The \textit{diagonal section} of a copula $C$ is a function $\delta_{\,C}:[0,1]\rightarrow[0,1]$ given by $\delta_{\,C}(t)=C(t,t).$
\end{definition}

Since every copula $C$ is bounded by the Fr\'echet-Hoeffding bounds $\max\{u+v-1,0\}\leq C(u,v)\leq\min\{u,v\}$ then $\max\{2t-1,0\}\leq\delta_{\,C}(t)\leq t.$ If $C=\Pi$ (independence copula) then $\delta_{\,\Pi}(t)=t^2.$ If $(X,Y)$ is a random vector of continuous random variables with underlying copula $C$ and $\text{PQD}(X,Y)$ or $\text{NQD}(X,Y)$ then $C(u,v)\geq uv$ or $C(u,v)\leq uv,$ respectively, for all $(u,v)$ in $[0,1]^2,$ and therefore $\delta_{\,C}(t)\geq t^2$ or $\delta_{\,C}(t)\leq t^2,$ respectively, for all $t$ in $[0,1].$ Hence, if there exist $t_1,t_2$ in $[0,1]$ such that $\delta_{\,C}(t_1)<t_1^2$ and $\delta_{\,C}(t_2)>t_2^2$ then neither $\text{PQD}(X,Y)$ nor $\text{NQD}(X,Y),$ and consequently this would imply that neither $\text{PRD}(Y|X)$ nor $\text{NRD}(Y|X).$ In case of this last scenario, this would not necessarily imply that a mean or median regression function $\mu(x)$ is non-monotone since Proposition \ref{prop:regresion} is a one-way implication, but at least raises the question and leads to propose and analyze break-point candidates. The following result is straightforward:

\begin{proposition}
  Let $C_1$ and $C_2$ be two copulas such that $C_1(u,v)\geq uv$ and $C_2(u,v)\leq uv$ for all $(u,v)\in[0,1]^2,$ and let $0<\theta<1.$ Then the diagonal section of the resulting gluing copula $C_{1,2,\theta}$ as in (\ref{eq:gluing}) satisfies
	\begin{equation}
	  \delta_{1,2,\theta}(t) \begin{cases}
		                         \,\geq t^2,  & \text{ if } \,0\leq t\leq\theta, \\
														 \,=\theta^2, & \text{ if } \,t = \theta, \\
														 \,\leq t^2,  & \text{ if } \,\theta\leq t\leq 1.
		                       \end{cases}
	\label{eq:gluingDiag}
	\end{equation}
\label{prop:gluingDiag}
\end{proposition}

Since the diagonal section $\delta_{\,C}$ of any copula $C$ is a continuous function, see \cite{Nel06}, we may choose and analyze as possible break-point candidates those where crossings between $\delta_{\,C}$ and $\delta_{\,\Pi}$ take place.

\begin{example}\label{ex:NelsenEj33bis2}
  Continuing with Example \ref{ex:NelsenEj33}, from formula (\ref{eq:NelsenEj33A}) the corresponding diagonal section is:
	\begin{equation}\label{diagEjemplo}
	  \delta_{\theta}(t)\,=\,C_{\theta}(t,t)\,=\,\begin{cases}
		                                             \,\theta t\,, & t\leq\frac{1}{2-\theta}\,, \\
																								 \,2t-1\,, & t\geq\frac{1}{2-\theta}\,.
		                                           \end{cases}
	\end{equation}
	If $0<t\leq\frac{1}{2-\theta}$ then $\delta_{\theta}(t)\geq t^2$ if and only if $t\leq\theta.$ If $\frac{1}{2-\theta}\leq t<1$ then $\delta_{\theta}(t)\leq t^2.$ Since $0<\theta<1$ then $\theta<\frac{1}{2-\theta}$ and therefore we conclude that $\delta_{\theta}(t)\geq t^2$ if and only if $t\leq\theta,$ and $\delta_{\theta}(t)\leq t^2$ if and only if $t\geq\theta.$ Hence, we would propose $t=\theta$ as break-point candidate, as expected. See Fig. \ref{fig:Examples23} (right).$\qquad_{\blacksquare}$
\end{example}

\begin{figure}[t]
  \sidecaption
  \includegraphics[scale=.30]{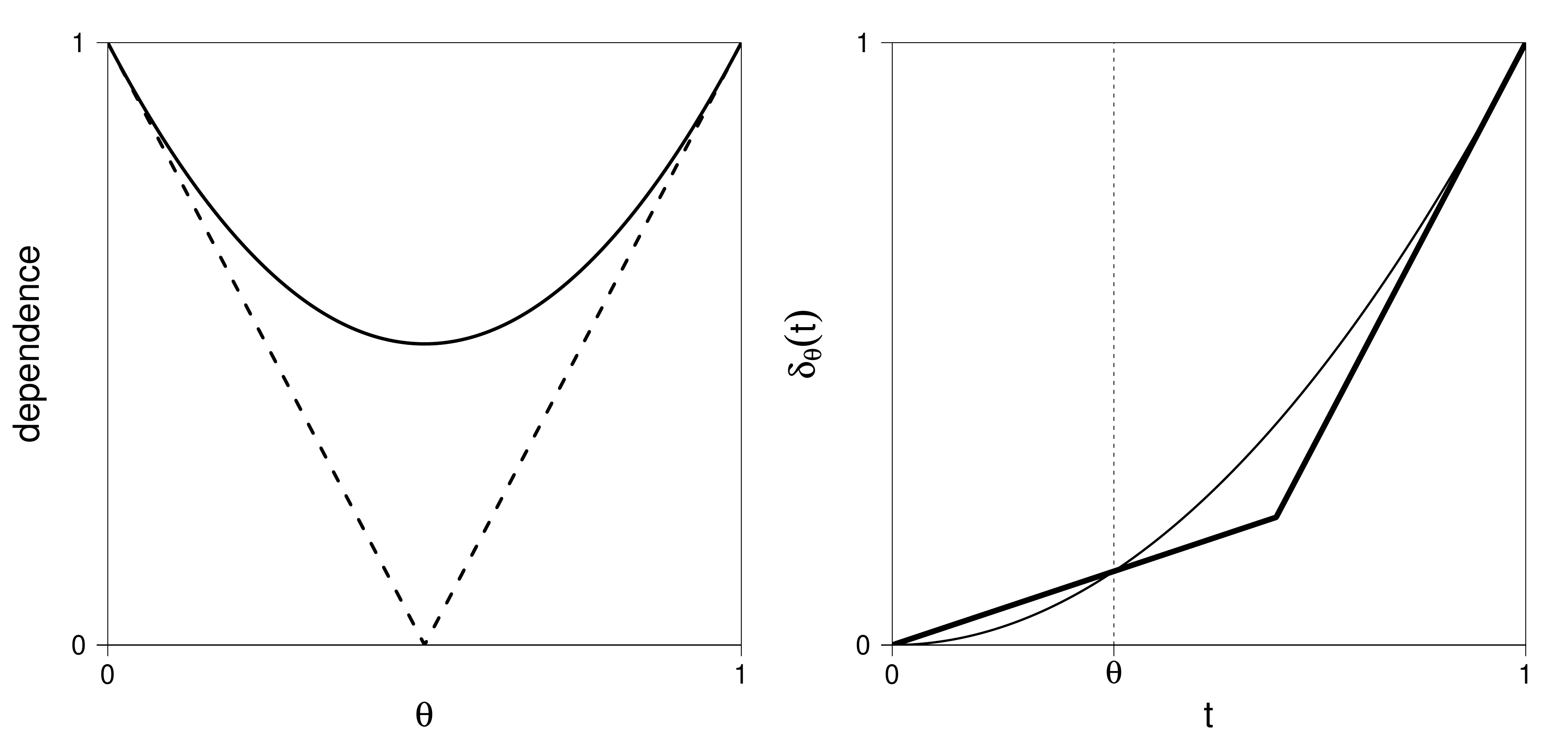}
  \caption{Left: $|\rho_{\theta}| $ (dashed line) and $\sigma_{\theta}$ (continuous line) in Example \ref{ex:NelsenEj33bis}. Right: $\delta_{\theta}$ (thick line) and $\delta_{\Pi}$ (thin line) in Example \ref{ex:NelsenEj33bis2}.}
  \label{fig:Examples23}
\end{figure}

\begin{example}\label{Dette}
  This is one of the examples used in \cite{DetHecVol14} to raise concerns about the use of copulas when the dependence relationship between random variables implies a non-monotone regression function. Let $\varepsilon$ be a Normal$(0,1)$ random variable, a constant $k^2 = 0.01,$ and $X$ a Uniform$(0,1)$ random variable independent from $\varepsilon.$ Now define the random variable:
	\begin{equation}\label{eq:Yreg}
	  Y\,=\,(X\,-\,0.5)^2\,+\,k\varepsilon\,.
	\end{equation}
	Then the conditional distribution of $Y$ given $X=x$ is Normal$\big((x-0.5)^2,k^2\big)$ and therefore the corresponding mean regression function is given by:
	\begin{equation}\label{eq:YregFunc}
	  \mu(x)\,=\,\mathbb{E}(Y\,|\,X=x)\,=\,(x\,-\,0.5)^2\,,\quad 0\leq x\leq 1,
	\end{equation}
clearly a non-monotone regression function (decreasing when $x\leq 0.5,$ increasing when $x\geq 0.5).$ Since the joint probability density of $(X,Y)$ is given by $f_{X,Y}(x,y)=f_X(x)f_{Y|X}(y\,|\,x)$ then:
	\begin{eqnarray}\label{eq:FXY}
	  F_{X,Y}(x,y) &=& \int_{-\infty}^{\,x}f_X(r)\int_{-\infty}^{\,y}f_{Y|X}(s\,|\,x)\,dsdr = \int_{-\infty}^{\,x}f_X(r)F_{Y|X}(y\,|\,r)\,dr \nonumber \\
		             &=& \begin{cases}
								        \,0\,, & \text{ if } x\leq 0, \\
								        \displaystyle{\int_0^{\,x}\Phi\Big(\frac{y-(r-0.5)^2}{k}\Big)\,dr\,,} & \text{ if } 0<x<1, \\
												\displaystyle{\int_0^1\Phi\Big(\frac{y-(r-0.5)^2}{k}\Big)\,dr\,,} & \text{ if } x\geq 1,
										 \end{cases}
	\end{eqnarray}
\noindent where $\Phi$ is the distribution function for a Normal$(0,1)$ random variable. From (\ref{eq:FXY}) it is possible obtain the following expression for the marginal distribution function of $Y:$
	\begin{equation}\label{eq:FY}
	  F_Y(y) \,=\, F_{X,Y}(+\infty,y) \,=\, \int_0^1\Phi\Big(\frac{y-(r-0.5)^2}{k}\Big)\,dr.
	\end{equation}
Hence, by Sklar's Corollary 2.3.7 in \cite{Nel06} it is possible to obtain the following expression for the underlying copula of $(X,Y):$
  \begin{equation}\label{eq:copula}
	  C(u,v) \,=\, F_{X,Y}\big(F_X^{(-1)}(u), F_Y^{(-1)}(v)\big) \,=\, \int_0^{\,u}\Phi\Big(\frac{F_Y^{-1}(v)-(r-0.5)^2}{k}\Big)\,dr\,,
	\end{equation}
and consequently the diagonal section of such copula is given by:
  \begin{equation}\label{eq:diagonal}
	  \delta_C(t) \,=\, C(t,t) \,=\, \int_0^{\,t}\Phi\Big(\frac{F_Y^{-1}(t)-(r-0.5)^2}{k}\Big)\,dr.
	\end{equation}

\begin{figure}[t]
  \sidecaption
  \includegraphics[scale=.30]{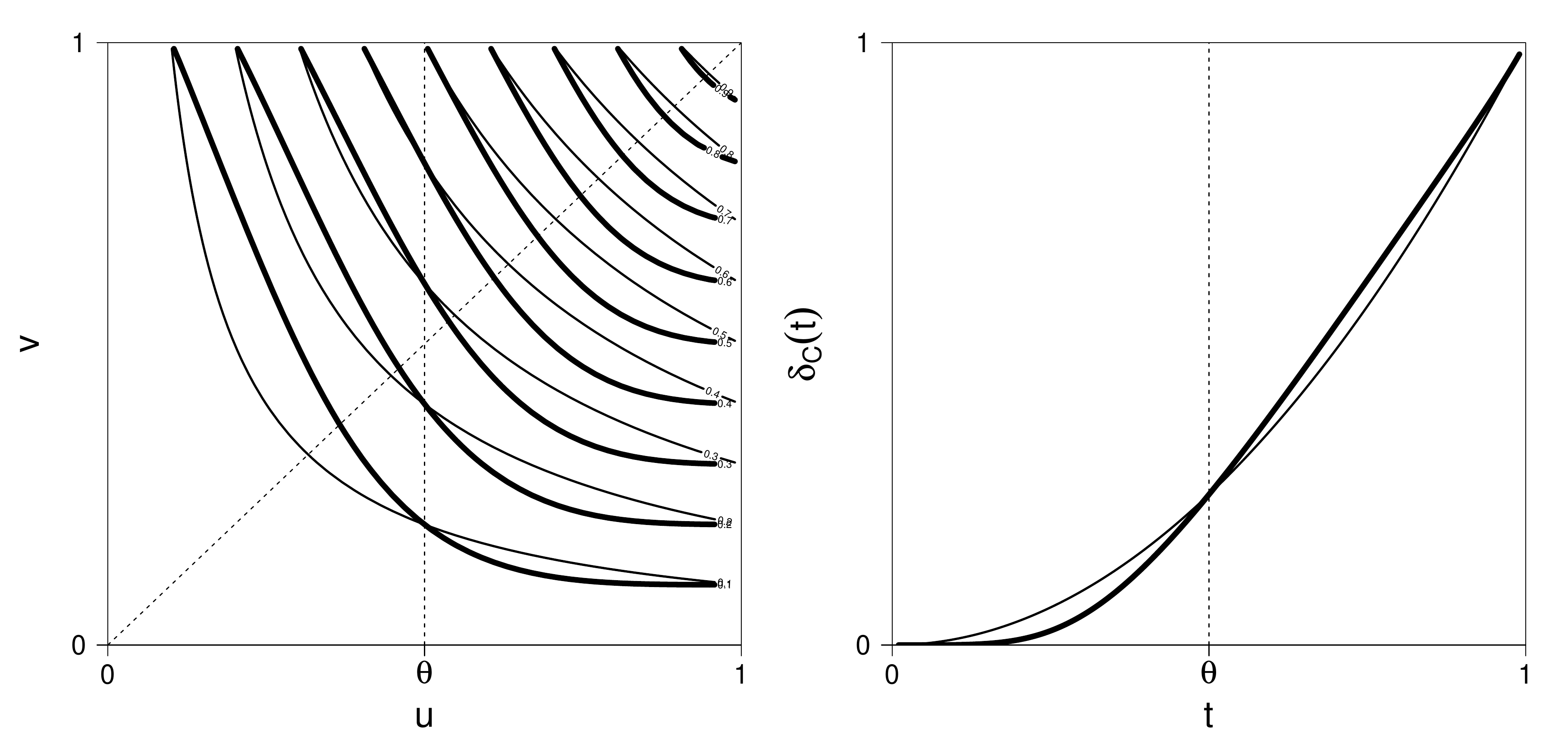}
  \caption{Example \ref{Dette}. Left: Level curves (thick style) of copula (\ref{eq:copula}) versus level curves (thin style) of product (or independence) copula. Right: Diagonal section (thick style) of copula (\ref{eq:copula}) versus diagonal section (thin style) of product (or independence) copula.}
  \label{fig:Example4A}
\end{figure}

\noindent In Fig. \ref{fig:Example4A} (left) we may notice crossings between copula (\ref{eq:copula}) level curves (thick style) and the product (or independence) copula $\Pi(u,v)=uv$ level curves (thin style), with the following interpretation: thick curve below thin curve implies $C(u,v)\geq\Pi(u,v)$ and thick curve above thin curve implies $C(u,v)\leq\Pi(u,v).$ In Fig. \ref{fig:Example4A} (right) the graph of the diagonal section (\ref{eq:diagonal}) is compared to the graph of the diagonal section of $\Pi$ from where we get as gluing point candidate $u=\theta=1/2.$

\medskip

\noindent Then we proceed to a \textit{gluing copula decomposition} by means of (\ref{eq:gluing}) where $C_{1,2,\theta}=C.$ For $0\leq u\leq\theta$ we get $\theta C_1(\frac{u}{\theta},v)=C(u,v),$ and if we let $u_*=\frac{u}{\theta}\in[0,1]$ then:
  \begin{equation}\label{eq:copula1}
    C_1(u_*,v) \,=\, \frac{1}{\theta}C(\theta u_*,v) \,=\, 2\int_0^{\,u_*/2}\Phi\Big(\frac{F_Y^{-1}(v)-(r-0.5)^2}{k}\Big)\,dr\,,
  \end{equation}
and therefore:
  \begin{equation}\label{eq:parcialcopula1}
	  \frac{\partial}{\partial u_*}C_1(u_*,v) \,=\, \Phi\Big(\frac{F_Y^{-1}(v)-0.25(1-u_*)^2}{k}\Big)\,,
	\end{equation}
where clearly (\ref{eq:parcialcopula1}) is a non-decreasing function of $u_*$ which by Corollary \ref{cor:NelsenThm5210} implies NRD for copula $C_1,$ and consequently NQD by Corollary \ref{cor:QDimpliesRD}. Also, by Proposition \ref{prop:regresion} we get that a regression function $\mu_1(x)$ based on $C_1$ will lead to a non-increasing function of $x.$ See Fig. \ref{fig:Example4B} (left) for the level curves of $C_1$ (thick style) versus the level curves (thin lines) of $\Pi(u,v)=uv,$ where all the level curves of $C_1$ are above the corresponding ones to $\Pi$ implying that $C_1(u,v)\leq\Pi(u,v),$ as expected.

\begin{figure}[t]
  \sidecaption
  \includegraphics[scale=.30]{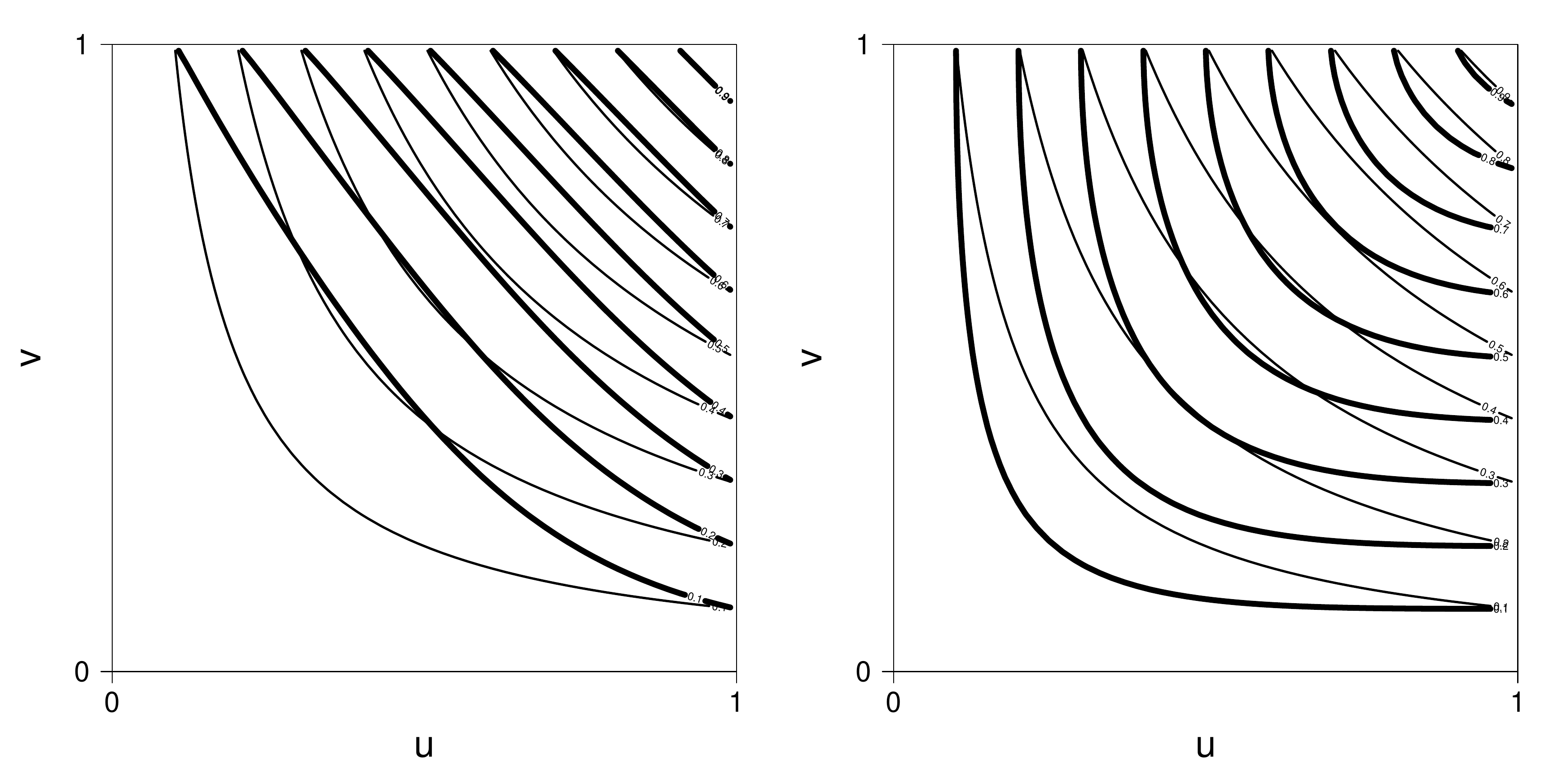}
\caption{Example \ref{Dette}. Left: Level curves (thick style) of copula (\ref{eq:copula1}) versus level curves (thin style) of product (or independence) copula. Right: Level curves (thick style) of copula (\ref{eq:copula2}) versus level curves (thin style) of product (or independence) copula.}
  \label{fig:Example4B}
\end{figure}

\medskip
\noindent Similarly, for $\theta\leq u\leq 1$ we get $(1-\theta)C_2(\frac{u-\theta}{1-\theta},v) + \theta v=C(u,v)$ and if we let $u_*=\frac{u-\theta}{1-\theta}\in[0,1]$ then:
  \begin{eqnarray}
    C_2(u_*,v) &=& \frac{C((1-\theta)u_*+\theta,v)-\theta v}{1-\theta} \nonumber \\
		           &=& 2\int_0^{(u_*+1)/2}\Phi\Big(\frac{F_Y^{-1}(v)-(r-0.5)^2}{k}\Big)\,dr\,-\,v\,,
  \end{eqnarray}\label{eq:copula2}
and therefore:
  \begin{equation}\label{eq:parcialcopula2}
	  \frac{\partial}{\partial u_*}C_2(u_*,v) \,=\, \Phi\Big(\frac{F_Y^{-1}(v)-0.25u_*^2)}{k}\Big)\,,
	\end{equation}
where clearly (\ref{eq:parcialcopula2}) is a non-increasing function of $u_*$ which by Corollary \ref{cor:NelsenThm5210} implies PRD for copula $C_2,$ and consequently PQD by Corollary \ref{cor:QDimpliesRD}. Also, by Proposition \ref{prop:regresion} we get that a regression function $\mu_2(x)$ based on $C_2$ will lead to a non-decreasing function of $x.$ See Fig. \ref{fig:Example4B} (right) for the level curves of $C_2$ (thick style) versus the level curves (thin lines) of $\Pi(u,v)=uv,$ where all the level curves of $C_2$ are below the corresponding ones to $\Pi$ implying that $C_2(u,v)\geq\Pi(u,v),$ as expected.

\medskip

\noindent In summary, the dependence between $X$ and $Y$ induced by (\ref{eq:Yreg}), which by construction has a regression function $\mu(x)$ that is non-monotone, has an underlying copula $C$ given by (\ref{eq:copula}) with a diagonal section $\delta_C$ given by (\ref{eq:diagonal}) that gives as gluing point candidate $\theta=1/2,$ leading to a gluing copula decomposition as in (\ref{eq:gluing}) where $C_1$ is NQD and NRD and therefore leads to a non-increasing regression function $\mu_1(x),$ and where $C_2$ is PQD and PRD and therefore leads to a non-decreasing regression function $\mu_2(x),$ that is:
  \begin{equation}\label{eq:nomonotreg}
	  \mu(x) \,=\, \begin{cases}
		               \,\mu_1(x)\,\downarrow\,, & u=F_X(x)=x\leq\theta=1/2\,, \\
									 \,\mu_2(x)\,\uparrow\,, & u=F_X(x)=x\geq\theta=1/2.
		             \end{cases}
	\end{equation}
In this example it was possible to obtain a gluing copula decomposition as in (\ref{eq:gluing}) of the underlying copula $C$ into $C_1$ and $C_2$ being these last two copulas NQD and PQD, respectively, and therefore candidates to be approximated by well known totally ordered families of copulas.$\qquad_{\blacksquare}$
\end{example}

\section{Final remarks}
\label{sec:Finrmks}

If $(X,Y)$ is a bivariate random vector of continuous random variables with an underlying copula $C$ such that $|\rho_C|<\sigma_C$ then $C$ is neither PQD nor NQD and therefore neither PRD nor NRD. Many of well known parametric families of copulas are totally ordered (that is, PQD and/or NQD) and in such case they have to be discarded as admissible copulas for $(X,Y).$ To face this challenge, in the present work it has been proposed a \textit{gluing copula decomposition} of $C$ into totally ordered copulas that combined may lead to a non-monotone regression function.

\begin{acknowledgement}
The present work was partially supported by project IN115817 from Programa de Apoyo a Proyectos de Investigaci\'on e Innovaci\'on Tecnol\'ogica (PAPIIT) at Universidad Nacional Aut\'onoma de M\'exico.
\end{acknowledgement}

\bibliographystyle{spmpsci}
\bibliography{ArturoErdely_biblio}

\end{document}